\documentclass[12pt,a4paper]{amsart}
\usepackage{amsmath,amssymb,amsfonts,amsthm}
\usepackage[hypertex]{hyperref} 
\usepackage[mathscr]{eucal}
\usepackage{comment}

\date{1 (14) December  2015}

\author{Theodore~Voronov}
\address{School of Mathematics, University of Manchester, Manchester, M60 1QD, UK\\
{\hphantom{hh}Dept. of Quantum Field Theory, Tomsk State University, Tomsk, 634050, Russia}}
\email{theodore.voronov@manchester.ac.uk}

\title[Quantum microformal  morphisms]{Quantum microformal  morphisms of supermanifolds: an explicit formula  and further properties}

\newtheorem{theorem}{Theorem}

\theoremstyle{definition}
\newtheorem{definition}{Definition}
\theoremstyle{remark}
\newtheorem*{remark}{Remark}

\def\co{\colon\thinspace}

\renewcommand{\geq}{\geqslant}

\newcommand{\der}[2]{{\frac{\partial {#1}}{\partial {#2}}}}

\newcommand{\fun}{C^{\infty}}
\DeclareMathOperator{\ofun}{\mathit{OC_{\hbar}^{\infty}}}

\newcommand{\f}{{\varphi}}

\newcommand{\itt}{{\tilde\imath}}

\newcommand{\tto}{{\linethickness{2pt}
		  \,\begin{picture}(1,0)
                   \put(0,0.26){\line(1,0){0.95}}
                   \put(0,0){$\boldsymbol{\rightarrow}$}
                  \end{picture}
                  }\,
}

\newcommand{\ttoq}{\tto_q}

\begin{document}
\begin{abstract} We give an explicit formula, as a formal differential operator, for quantum microformal morphisms of (super)manifolds that we introduced earlier. Such quantum microformal morphisms are essentially oscillatory integral operators or Fourier integral operators of a particular kind. They act on oscillatory wave functions, whose algebra extends the algebra of formal power series in Planck's constant.

In the classical limit, quantum microformal morphisms reproduce, as the main term of the asymptotic, the nonlinear pullbacks of functions with respect to `\,classical\,' microformal morphisms. We found these nonlinear transformations  of functions in search of  $L_{\infty}$-morphisms for homotopy Poisson structures.
\end{abstract}

\maketitle

In our previous paper~\cite{tv:oscil}, we introduced `quantum' microformal morphisms between (super)manifolds. They are essentially oscillatory integral operators or Fourier integral operators of a particular kind acting on functions on one (super)manifold and transforming them to functions on another (super)manifold. Here we give an explicit formula for them, expressing them as   formal differential operators of infinite order. (There is some mild resemblance of the form of this formula with the expression  for a partition function for interacting fields in QFT.)

The motivation for considering such objects is as follows. 

We first discovered `\,classical\,' microformal morphisms~\cite{tv:nonlinearpullback, tv:microformal} \,---\, in search of a construction of $L_{\infty}$-morphisms for homotopy Poisson structures on manifolds. They arise as an amazing generalization of ordinary smooth maps, for which there
is a notion of pullback of functions generalizing   ordinary pullbacks. Unlike the latter, pullbacks with respect to microformal morphisms
are    nonlinear (formal) transformations of spaces of functions. Being nonlinear, they cannot be ring homomorphisms as ordinary pullbacks. However, they possess a remarkable property that their derivatives at each point are ring homomorphisms.

`\,Quantum\,'  microformal morphisms~\cite{tv:oscil} are in the same relation to `\,classical\,' microformal morphisms~\cite{tv:nonlinearpullback, tv:microformal} as the Schr\"{o}dinger equation is to the Hamilton--Jacobi equation. Classical microformal morphisms are constructed by methods of symplectic geometry, in particular,  generating functions.  Equations of Hamilton--Jacobi form arise there at many instances when we connect this theory with (homotopy) Poisson structures. We felt from the start that this must be a projection of a hypothetical quantum version, which was eventually found in~\cite{tv:oscil}.

Let us give precise definitions and statements.

Quantum   microformal morphisms act on `oscillatory wave functions'. 
\begin{definition} An \emph{oscillatory wave function} on a (super)manifold $M$ is a linear combination of  formal expressions of the form
\begin{equation*}
    w(x)=a(\hbar, x)e^{\frac{i}{\hbar}b(x,\hbar)}
\end{equation*}
where $a(x,\hbar)$ and $b(x,\hbar)$ are  formal power series in $\hbar$ whose coefficients are smooth functions on $M$. By writing $b(x,\hbar)=b_0(x)+\hbar b_1(x)+\ldots$ and taking logarithms, we can always re-write the exponential so that    $w(x)=A(\hbar, x)e^{\frac{i}{\hbar}b_0(x)}$, with    some formal  power series $A(x,\hbar)$. 

We arrive at the \emph{algebra of oscillatory wave functions} on $M$, which we denote $\ofun(M)$, so that $\ofun(M)=\fun(M)[[\hbar]]\exp \frac{i}{\hbar}\fun(M)$. It is an algebra over  the algebra $\fun(M)[[\hbar]]$ of formal power series in $\hbar$.
\end{definition}

Recall the construction of quantum   microformal morphisms~\cite{tv:oscil}. Consider supermanifolds $M_1$ and $M_2$. In  given coordinate systems on $M_1$ and $M_2$, a quantum   microformal morphism $\hat\Phi \co M_1\ttoq  M_2$ is specified by a quantum generating function $S_{\hbar}(x,q)$. Here $x^a$ are local coordinates on $M_1$, $y^i$ are local coordinates on $M_2$, and $p_a$ and $q_i$ denote the corresponding components of momenta. The function $S_{\hbar}(x,q)$ is a formal power series in $q_i$\,:
\begin{equation*}
    S_{\hbar}(x,q)=S_{\hbar}^0(x)+ \f^i_{\hbar}(x)q_i + \frac{1}{2}\,S^{ij}_{\hbar}(x)q_jq_i + \frac{1}{3!}\,S^{ijk}_{\hbar}(x)q_kq_jq_i + \ldots
\end{equation*}
The coefficients of the expansion are formal power series in $\hbar$. Note that $S_{\hbar}(x,q)$ is a coordinate-dependent object and not a scalar function. Its transformation law will be clarified later. (The reason why we denoted the coefficients of the linear term differently from the others will also become clear shortly.)

\begin{remark} There are two different kinds of power expansions in this theory: expansions in $\hbar$ and expansions in the ``coupling constant''. More precisely, the latter are the expansions in the momenta $q$ for generating functions $S_{\hbar}(x,q)$ and the corresponding  
expansions in powers of the functions and their derivatives in the formulas for pullbacks and the compositions of morphisms.
\end{remark}

\begin{definition} By definition, a   \emph{quantum   microformal morphism} $\hat\Phi\co M_1\ttoq  M_2$ is identified with its action on functions (in the opposite direction), called \emph{quantum pullback} and denoted $\hat\Phi^*$, which is defined  by the formula
\begin{equation}\label{eq.phihat}
    (\hat\Phi^* w)(x)=\frac{1}{(2\pi\hbar)^{n_2}} \int_{T^*M_2} D(y,q) \,\, e^{\frac{i}{\hbar}(S_{\hbar}(x,q)-y^iq_i)}\,w(y)\,.
\end{equation}
The integral operator $\hat\Phi^*$ is a linear transformation $\hat\Phi^*\co \ofun(M_2)\to \ofun(M_1)$ of oscillatory wave functions. 
\end{definition}
(Integration in~\eqref{eq.phihat} is w.r.t. the Liouville measure on $T^*M_2$. The coefficient in front of the integral is the standard factor depending only on  dimension. For the simplicity of notation, it is written here for the more familiar even case. Everything works for supermanifolds.)

In the quasi-classical limit $\hbar\to 0$, a quantum   microformal morphism $\hat\Phi\co M_1\ttoq  M_2$ becomes a classical microformal morphism $\Phi\co M_1\tto   M_2$ as defined in~\cite{tv:nonlinearpullback, tv:microformal}. More precisely:

\begin{theorem}[\cite{tv:oscil}] Consider the `classical generating function' $S_0(x,q)$ such that
 \begin{equation*}
    S_{\hbar}(x,q)=S_0(x,q)+O(\hbar)\,.
 \end{equation*}
Let $\hat\Phi$ be the the classical microformal morphism defined by  $S_0(x,q)$.  Then for an arbitrary oscillatory wave function of the form $w(y)=e^{\frac{i}{\hbar}g(y)}$ on $M_2$, the quantum pullback $\hat\Phi^*(w)$ 
is an oscillatory wave function $e^{\frac{i}{\hbar}f_{\hbar}(x)}$ on $M_2$ such that
\begin{equation*}
    f_{\hbar}=\Phi^*(g)+O(\hbar)\,,
\end{equation*}
where  $\Phi^*(g)$ is the  pullback   of $g$    w.r.t. the microformal morphism $\Phi\co M_1\tto   M_2$  as  defined in~\cite{tv:nonlinearpullback, tv:microformal}.
\end{theorem}

{\small

\begin{proof}
The quantum pullback $\hat\Phi^*w$ of an oscillatory wave function $w(y)=e^{\frac{i}{\hbar}g(y)}$ is given by the integral w.r.t. $y,q$ of the oscillating exponential
\begin{equation*}
    e^{\frac{i}{\hbar}(S_{\hbar}(x,q)+g(y)-y^iq_i)}\,.
\end{equation*}
By the stationary phase method, the value of the integral, in the main term in $\hbar$, is the same exponential evaluated at the critical points of the phase when $\hbar\to 0$. By differentiating w.r.t. $y^i$ and $q_i$ and setting the result to zero, we obtain the system of equations
 \begin{equation*}
    q_i=\der{g}{y^i}(y)\,,\quad y^i=(-1)^{\itt}\der{S_0}{q_i}(x,q) 
 \end{equation*}
for determining $y^i, q_i$, the unique solution of which should be substituted into $S_{0}(x,q)+g(y)-y^iq_i$. This is exactly the construction of $f=\Phi^*(g)$ in~\cite{tv:nonlinearpullback, tv:microformal}.
\end{proof}

}

(Note that $\Phi^*\co \fun(M_2)\to \fun(M_1)$ is,  in general, a nonlinear formal transformation of the spaces of functions. It is the usual pullback, $\Phi^*=\f^*$, if $\Phi=\f$ is an ordinary map $\f\co M_1\to M_2$. In such case, $S(x,q)=\f^i(x)q_i$, where $y^i=\f^i(x)$.)

Surprisingly, the quantum pullback $\hat \Phi^*\co \ofun(M_2)\to \ofun(M_1)$ can be expressed in a closed form as a formal differential operator, i.e., the integral~\eqref{eq.phihat} can be  explicitly evaluated. (This is an advantage over the corresponding classical pullback, which is in general known only in terms of an iterative procedure.) It is actually very simple.

Let us write a quantum generating function $S_{\hbar}(x,q)$ defining a quantum microformal morphism $\hat\Phi\co M_1\ttoq M_2$ as
\begin{equation}\label{eq.qgenfun}
    S_{\hbar}(x,q)=S_{\hbar}^0(x)+\f^i_{\hbar}(x)q_i+S^{+}_{\hbar}(x,q)\,,
\end{equation}
where   $S^{+}_{\hbar}(x,q)$ is the sum of all  terms of order $\geq 2$ in $q_i$. Such grouping of the terms of the expansion in $q_i$ is explained by the result below. 

\begin{theorem} The action of $\hat \Phi^*$ on oscillatory wave functions can be expressed as follows:
\begin{equation}\label{eq.phiw}
     \bigl(\hat \Phi^*w\bigr)(x)=e^{\frac{i}{\hbar}S_{\hbar}^0(x)}
    \left(e^{\frac{i}{\hbar}S^{+}_{\hbar}\left(x,\frac{\hbar}{i}\der{}{y}\right)}w(y)\right)_{\left|\vphantom{\int\limits_a^b}\ y^i=\f^i_{\hbar}(x)\right.}\,.
\end{equation}
That is, it is the combination of a formal differential operator in $y^i$ applied to $w(y)$ followed by the substitution $y^i=\f^i_{\hbar}(x)$ and the multiplication by a given phase factor on $M_1$. 
\end{theorem}
\begin{proof} We have 
\begin{multline*} 
    (\hat\Phi^* w)(x)=\frac{1}{(2\pi\hbar)^{n_2}} \int D(y,q) \,\, e^{\frac{i}{\hbar}(S_{\hbar}^0(x)+\f^i_{\hbar}(x)q_i+S^{+}_{\hbar}(x,q)-y^iq_i)}\,w(y)=\\
    e^{\frac{i}{\hbar}S_{\hbar}^0(x)}\frac{1}{(2\pi\hbar)^{n_2}} \int  Dq \,\, 
    e^{\frac{i}{\hbar}\f^i_{\hbar}(x)q_i}
    e^{\frac{i}{\hbar}S^{+}_{\hbar}(x,q)}
    \int Dy \,\, 
    e^{-\frac{i}{\hbar}y^iq_i}
    \,w(y)\,.
\end{multline*}
The integral  is the composition of the Fourier transform of $w(y)$ to $q_i$, the multiplication by a function of $q_i$, and the inverse Fourier transform from $q_i$ to $y^i$, where $\f^i_{\hbar}(x)$ is substituted for $y^i$. Recalling the standard relation between multiplication and differentiation under Fourier transform, we arrive  at the claimed result.
\end{proof}

Taking into account that $S^{+}_{\hbar}(x,q)$ contains  terms of order $2$ and higher in $q_i$, we see that in the exponential in~\eqref{eq.phiw} there is a formal differential operator (of infinite order) of the form
\begin{equation*}
    \frac{\hbar}{i} \frac{1}{2}\,S^{ij}_{\hbar}(x)\der{}{y^i}\der{}{y^i} + \ldots
\end{equation*}
where the dots stand for terms with the higher derivatives and at the same time of higher order in $\hbar$. Therefore the exponential of it will be a formal differential operator in $y^i$ of the form $1+O(\hbar)$, where the differential part starts from second derivatives. Presumably, any such operator can be conversely expressed as an exponential as above. This more or less describes the `differential' part of a quantum pullback. We may conclude that a quantum microformal morphism consists of the three ingredients: a formal $\hbar$-dependent  differential operator on $M_2$ as described, an actual smooth map, or its quantum perturbation, $\f_{\hbar}\co M_1\to M_2$, and an $\hbar$-dependent phase factor on $M_1$. What looks a bit unsatisfactory in such a picture, is that the differential operators in $y^i$ appearing there are `with constant coefficients' (independent of $y^i$), which obviously seems coordinate-dependent. However, this difficulty is only imaginary; it is a problem of a description, not a problem with the definition of our objects. We may see this as follows.

If we first perform the integration over $q_i$ in~\eqref{eq.phihat}, we will express $\hat\Phi^*$ as an integral operator with a (distributional) integral kernel $K(x,y)$. This kernel is nothing but the Fourier transform  of the phase function
\begin{equation*}
    e^{\frac{i}{\hbar}(S_{\hbar}^0(x)+\f^i_{\hbar}(x)q_i+S^{+}_{\hbar}(x,q))}
\end{equation*}
(from $q_i$ to $y^i$). One can easily see (this is very close to the argument  above), that such a Fourier transform can be represented, apart from the phase factor $e^{\frac{i}{\hbar}S_{\hbar}^0(x)}$, which is just an oscillatory wave function on $M_1$, as the action of the differential operator
\begin{equation*}
    e^{\frac{i}{\hbar}S^{+}_{\hbar}\left(x,\frac{\hbar}{i}\der{}{y}\right)}
\end{equation*}
on  the delta-function $\delta(y-\f_{\hbar}(x))$. This is of course almost the same as above, but now it gives a clear geometric picture: the Schwarz kernel $K(x,y)$ of the operator $\hat\Phi^*$  is a distribution supported on the graph of the `map' $\f_{\hbar}\co M_1\to M_2$. This is well-defined and independent of the choice of coordinates. 

In the description of $\hat\Phi^*$, it would therefore  be better not to separate the action of the differential operator and the substitution $y^i=f^i_{\hbar}(x)$, but treat them together as a kind of `\,differential operator over a map $\f_{\hbar}\co M_1\to M_2$\,'. 

With the same argument we can   address the questions concerning the transformation laws for our generating functions. A quantum generating function $S_{\hbar}(x,q)$ can be regarded as a genuine function w.r.t. $x$; as for its dependence on $y$, the transformation law is as follows. For the exponential 
\begin{equation*}
    e^{\frac{i}{\hbar}S_{\hbar}(x,q)}\,,
\end{equation*}
the transformation under a change of coordinates $y=y(y')$   is the composition of the Fourier transform from the variables $q_i$ to the variables $y^i$, the usual change of variables in the resulting function of $y^i$, and the inverse Fourier transform from the variables $y^{i'}$ to the variables $q_{i'}$. Note that the first Fourier transform gives a distribution supported on the graph of $\f_{\hbar}\co M_1\to M_2$ (such as the delta-function $\delta(y-\f_{\hbar}(x))$ and its derivatives); there is no difficulty in performing invertible changes of variables  w.r.t. $y^i$ in such distributions. The Fourier transforms in question are well-defined and give a one-to-one correspondence between generating functions and distributions in this class. The described transformation law defines quantum generating functions $S_{\hbar}(x,q)$ as geometric objects on $M_1\times M_2$. (Recalling  that   Legendre transform may be regarded as the classical limit of Fourier transform, from here we may  obtain again the transformation law for classical generating functions $S_{\text{class}}(x,q)$ as established in~\cite{tv:nonlinearpullback}.)


\begin{thebibliography}{1}


\bibitem{tv:nonlinearpullback}
Th.~Th. Voronov.
\newblock The ``nonlinear pullback'' of functions and a formal category
  extending the category of supermanifolds.
\newblock \texttt{arXiv:1409.6475 [math.DG]}.

\bibitem{tv:microformal}
Th.~Th. Voronov.
\newblock Microformal geometry.
\newblock \texttt{arXiv:1411.6720 [math.DG]}.

\bibitem{tv:oscil}
Th.~Th. Voronov.
\newblock Thick morphisms of supermanifolds and oscillatory integral operators.
\newblock \texttt{arXiv:1506.02417 [math.DG]}.

\end{thebibliography}
\def\cprime{$'$}

\end{document}